\documentclass[twoside,runningheads]{llncs}

\usepackage[latin1]{inputenc}
\usepackage{enumerate}
\usepackage[hang,flushmargin]{footmisc}
\usepackage{amsmath}
\usepackage{amsfonts}
\usepackage{amssymb}
\usepackage{mathtools}
\usepackage{textcomp}
\usepackage{stmaryrd}
\usepackage{bm}
\usepackage{proof}
\usepackage{tikz}
\usetikzlibrary{shapes.geometric,arrows,positioning}

\usepackage{appendix}
\usepackage{tabularx}
\usepackage{booktabs}
\usepackage{xspace}
\usepackage{xcolor}
\usepackage{subfigure}

\usepackage{proof-at-the-end}

\usepackage{float}
\floatstyle{ruled}
\newfloat{model}{tbp}{model}
\floatname{model}{Model}
\newcounter{modelline}
\makeatletter
\newcommand{\mline}[1]{{\refstepcounter{modelline}\ltx@label{#1}}\quad\text{\scriptsize{\themodelline}}\quad}
\makeatother

\newcommand{\dL}{\textsf{dL}\xspace}
\newcommand{\dRL}{\textsf{dRL}\xspace}
\newcommand{\KeYmaeraX}{KeYmaera~X\xspace}

\newcommand{\passign}[2]{\ensuremath{#1 \coloneqq #2}}
\newcommand{\ptest}[1]{\ensuremath{?#1}}
\newcommand{\D}[1]{\ensuremath{#1{'}}}
\newcommand{\syssep}{\ensuremath{,\,}}
\newcommand{\pevolvein}[2]{\ensuremath{\{#1~\&~#2\}}}
\newcommand{\dbox}[2]{\ensuremath{\left[#1\right]#2}}
\newcommand{\didia}[2]{\ensuremath{\langle#1\rangle#2}}
\newcommand{\pchoice}[2]{\ensuremath{#1\,\cup\,#2}}
\newcommand{\prepeat}[1]{\ensuremath{#1^\ast}}
\newcommand{\ltrue}{\top}
\newcommand{\lfalse}{\bot}

\newcommand{\vars}{\mathcal{V}}
\newcommand{\reals}{\mathbb{R}}

\newcommand{\aminBrake}{\ensuremath a_\text{minBrake}}
\newcommand{\amaxBrake}{\ensuremath a_\text{maxBrake}}
\newcommand{\amaxAccel}{\ensuremath a_\text{maxAccel}}

\newcommand{\lbisubjunct}{\ensuremath{\leftrightarrow}}
\newcommand{\limply}{\ensuremath{\rightarrow}}

\newcommand{\boundvars}[1]{\ensuremath{\textsf{BV}(#1)}}
\newcommand{\stold}{\ensuremath{\omega}}
\newcommand{\stnew}{\ensuremath{\nu}}
\newcommand{\imodels}[2]{#1 \models #2}
\newcommand{\reachset}[1]{\ensuremath\llbracket #1 \rrbracket}
\newcommand{\irelmodels}[3]{(#1,#2) \in \reachset{#3}}
\newcommand{\valuein}[2]{\ensuremath#1\llbracket#2\rrbracket}
\newcommand{\truthset}[1]{\ensuremath\llbracket#1\rrbracket}

\newcommand{\pycomp}[1]{\ensuremath\textsf{Py}(#1)}
\newcommand{\pyrule}[1]{\textsc{#1}}
\newcommand{\pyeval}[2]{\ensuremath\textsf{eval}(#1)_{#2}}
\newcommand{\pybooleval}[2]{\ensuremath#1(\textsf{bool}~#2)}

\newcommand{\dethp}{\texttt{det-HP}\xspace}

\usepackage{prettyref}
\newcommand{\rref}[2][]{\prettyref{#2}}
\newrefformat{sec}{Section\,\ref{#1}}
\newrefformat{app}{Appendix\,\ref{#1}}
\newrefformat{def}{Def.\,\ref{#1}}
\newrefformat{thm}{Theorem\,\ref{#1}}
\newrefformat{prop}{Proposition\,\ref{#1}}
\newrefformat{lem}{Lemma\,\ref{#1}}
\newrefformat{cor}{Corollary\,\ref{#1}}
\newrefformat{rem}{Remark\,\ref{#1}}
\newrefformat{ex}{Example\,\ref{#1}}
\newrefformat{tab}{Table\,\ref{#1}}
\newrefformat{fig}{Fig.\,\ref{#1}}
\newrefformat{case}{case\,\ref{#1}}
\newrefformat{foot}{Footnote\,\ref{#1}}
\newrefformat{line}{Line\,\ref{#1}}
\newrefformat{model}{Model\,\ref{#1}}

\newcommand{\amax}{\amaxAccel}
\newcommand{\amin}{\aminBrake}
\newcommand{\amaxb}{\amaxBrake}

\title{Slow Down, Move Over: A Case Study in Formal Verification, Refinement, and Testing of the Responsibility-Sensitive Safety Model for Self-Driving Cars}
\titlerunning{Verification, Refinement, and Testing of RSS}
\author{Megan Strauss\orcidID{0009-0009-6769-2404} \and Stefan Mitsch\orcidID{0000-0002-3194-9759}}
\institute{
  Computer Science Department\\
  Carnegie Mellon University, Pittsburgh, USA\\
  \email{mstrauss@andrew.cmu.edu} \qquad \email{smitsch@cs.cmu.edu}
}

\begin{document}
\maketitle
\begin{abstract}
Technology advances give us the hope of driving without human error, reducing vehicle emissions and simplifying an everyday task with the future of self-driving cars. 
Making sure these vehicles are safe is very important to the continuation of this field. 
In this paper, we formalize the Responsibility-Sensitive Safety model (RSS) for self-driving cars and prove the safety and optimality of this model in the longitudinal direction. 
We utilize the hybrid systems theorem prover \KeYmaeraX to formalize RSS as a hybrid system with its nondeterministic control choices and continuous motion model, and prove absence of collisions. 
We then illustrate the practicality of RSS through refinement proofs that turn the verified nondeterministic control envelopes into deterministic ones and further verified compilation to Python. 
The refinement and compilation are safety-preserving; as a result, safety proofs of the formal model transfer to the compiled code, while counterexamples discovered in testing the code of an unverified model transfer back.
The resulting Python code allows to test the behavior of cars following the motion model of RSS in simulation, to measure agreement between the model and simulation with monitors that are derived from the formal model, and to report counterexamples from simulation back to the formal model.
\end{abstract}

\section{Introduction and Motivation}
\label{sec:introduction}
 As technology advances, it becomes both more appealing and feasible to automate everyday tasks. In doing this, there is a higher need for robust formal verification and testing strategies to ensure that these automated systems are safe. 
 Formal verification ensures that all situations encountered in a model of physics are safe, while testing provides evidence that the formal model works as expected and can give insight to bugs in a yet-to-be-verified model.
 To this end, a testing strategy analyzes some (representative) situations that an automated system may encounter.
 This paper develops a formal framework to link modeling, offline verification, online monitoring, and testing, sketched in \rref{fig:overview} and illustrated throughout the paper with an example in autonomous driving.
 
 \tikzstyle{style1} = [rectangle, rounded corners, minimum width=2cm, minimum height=1cm,text centered, draw=black, text width=2.5cm]
\tikzstyle{style2} = [rectangle, minimum width=2cm, minimum height=1.3cm, text centered, draw=black]

\tikzstyle{arrow} = [thick,->,>=stealth]
 \begin{figure}[htb]
 \centering
\begin{tikzpicture}
\node (formalize) [style1] {\footnotesize \textbf{Formalize HP}\\ \scriptsize Create HP $\alpha$};
\node (refinement) [style1, right=1.5cm of formalize] {\footnotesize \textbf{Refinement} \scriptsize Prove $D(\alpha) \leq \alpha$};
\node (python) [style1, right=2cm of refinement] {\footnotesize \textbf{Python} \\ \scriptsize Test $\pycomp{D(\alpha)}$};
\node (safety) [style2, below=.5cm of formalize,text width=3cm] {\footnotesize \textbf{Safety Proof (dL)} \\ \scriptsize $[\alpha]P$};
\node (transfer) [style2, below=.5cm of refinement,text width=3.7cm] {\footnotesize \textbf{Safety Transfer (dRL)}\\ \scriptsize $[\alpha]P \rightarrow [D(\alpha)]P$};
\node (compilation) [style2, below=.5cm of python,text width=3cm] {\footnotesize \textbf{Code Correctness} \tiny
\\ $(\pycomp{D(\alpha)}, \nu) \rightarrow (\perp, \omega)$ then $\irelmodels{\nu}{\omega}{D(\alpha)}$};

\draw [arrow] (formalize) -- node[anchor=south] {\tiny \parbox{1.5cm}{\centering deterministic $D(\alpha)$}} (refinement);
\draw [arrow] (refinement) -- node[anchor=south] {\tiny \parbox{2cm}{\centering compile \\$D(\alpha) \vartriangleright Py(D(\alpha))$}} (python);
\draw [arrow] (formalize) -- (safety);
\draw [arrow] (refinement) -- (transfer);
\draw [arrow] (python) -- (compilation);

\coordinate[below=.8cm of safety] (safetyarrowstart);
\coordinate[below=.8cm of compilation] (safetyarrowend);
\draw [-stealth, line width=2mm](safetyarrowstart) -- node[anchor=south]{\textbf{Safety Preservation}}(safetyarrowend);

\coordinate[below=.8cm of safetyarrowstart] (cexarrowend);
\coordinate[below=.8cm of safetyarrowend] (cexarrowstart);
\draw [-stealth, line width=2mm](cexarrowstart) -- node[anchor=south]{\textbf{Counterexample Transfer}}(cexarrowend);

\end{tikzpicture}
\caption{Proof and testing structure with relationships to each other}
\label{fig:overview}
\end{figure}
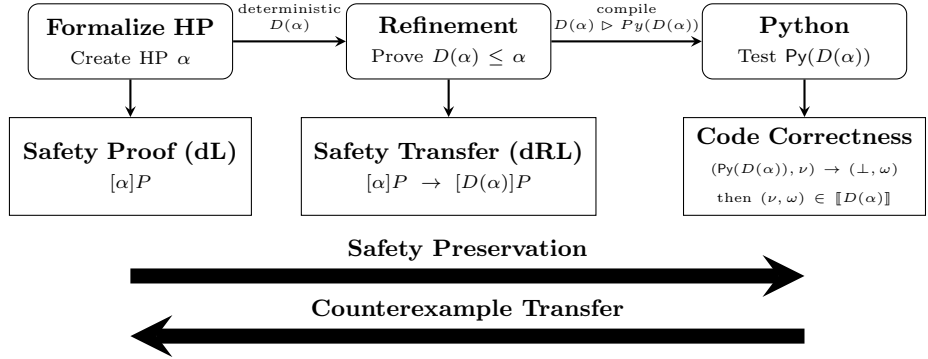

We use hybrid systems to describe the dynamic components of a car, with the goal of creating discrete computer controllers which prevent collisions and other unsafe situations. 
Analyzing such a system requires two components: a differential equation modeling the cars' motion and a car control algorithm. 
For the differential equation, we follow \cite{loos_safe_2011} and describe the longitudinal motion of a car in response to acceleration/deceleration inputs. 
For the discrete controller, we follow the Responsibility-Sensitive Safety (RSS) model \cite{shalev-shwartz_formal_2018}. 
Its discrete control choices influence the continuous motion of the car, and are given in terms of a nondeterministic control envelope that allows for many different ways of deterministic control implementation. 
In the resulting hybrid systems model, we prove the safety (defined by absence of collision) of all possible implementations within the control envelope. 
This proof step replaces the need to test through simulation whether the model is safe, i.e. checking if cars collide in any implementation of the system, but simulation can help in debugging unverified models. 
We also want the ability to analyze concrete (deterministic) controller implementations and motion models and determine whether they follow the general verified model. 
We do this through refinement proofs in differential refinement logic \cite{loos_platzer_2016} to a deterministic hybrid program followed by a verified compilation to Python which allows us to further analyze the behavior of a system. 
This way, we can analyze concrete behavior for concrete parameter choices in the formal model, and simulate behavior for analysis in select driving situations.
In order to transfer the safety proof to \emph{all encountered situations}, the remaining question after a safety proof is then to detect at runtime when the proof applies and when it does not. 
We answer this by applying ModelPlex \cite{mitsch_modelplex_2016} to generate monitors that are able to validate unverified code and motion model alike. 
We utilize the same monitors in testing and simulation to measure the robustness of control choices, or in other words, the safety-criticality of test scenarios. 
The proof, refinement, and compilation framework is set up in such a way that, if a proof step fails, compiling the formal model and analyzing it for counterexamples through testing is still possible.

\section{Background}

\subsection{Responsibility-Sensitive Safety}
We briefly summarize Responsibility-Sensitive Safety (RSS, \cite{shalev-shwartz_formal_2018}) that we use as a case study throughout this paper. 
The general purpose of RSS is to describe a safe algorithm for the interaction of (semi-)autonomous cars, with safety being defined as the avoidance of collision.
RSS defines a lane-based coordinate system which describes the position of cars with a longitudinal and lateral coordinate relative to the lane they are in on a road. 
The longitudinal coordinate measures the distance relative to a starting position along the road, which abstracts from road geometry. 
The lateral position describes the distance from the center of the road. 
The lane-based coordinate system simplifies the motion model into a linear differential equation system as in \cite{loos_safe_2011}.

\begin{table}[tb]
\caption{RSS variables and model parameters}
\label{tab:rssparams}
\begin{tabularx}{\linewidth}{l X} 
 \toprule
 Parameter & Description \\
 \midrule
 $x_1, x_2$ & positions in lateral or longitudinal direction of car 1 and 2\\
 $v_1, v_2$ & velocities in lateral or longitudinal direction of car 1 and car 2 \\
 $a_1, a_2$ & acceleration in lateral or longitudinal direction of car 1 and car 2 \\
  $a_{minBrake}$ & minimum deceleration during evasion maneuver  \\ 
 $a_{maxBrake}$ & maximum deceleration \\
 $a_{minBrakeCorrect}$ & minimum deceleration during evasion maneuver in the correct direction (only used for cars driving in opposite directions) \\
 $a_{maxAccel}$ & maximum acceleration (positive constant, used in both directions)\\
 $\rho$ & reaction time for starting evasion maneuver
\\
 \bottomrule
\end{tabularx}
\end{table}

The general idea behind the RSS control envelope is that two cars should maintain a safe distance between them, which is a function of the current speeds of the two cars. 
If at a safe distance, they are allowed to drive at any acceleration or deceleration within their cars' maximum acceleration and maximum braking bounds. 
Otherwise, if not at a safe distance, they must respond with a predefined evasion maneuver called \emph{proper response} within a specified reaction time horizon in order to remain safe.
The safe distance is specified as a function of (relative) vehicle speed and vehicle parameters, as summarized in \rref{tab:rssparams} \cite{shalev-shwartz_formal_2018}.

\subsection{Differential Dynamic Logic and Notation}
 We use differential dynamic logic (\dL) \cite{DBLP:journals/jar/Platzer17} to create models of car control and motion, and to express correctness and optimality properties of RSS. 
 Differential dynamic logic supports modeling hybrid systems complete with support for continuous dynamics expressed using differential equations. 
 Within the RSS model, we utilize discrete dynamics to describe the control options in terms of instantly changing the acceleration or deceleration that the car is following. 
 We use differential equations to model the resulting motion change in speed and position of the cars with respect to this chosen acceleration. 
 \dL comes with a proof calculus~\cite{DBLP:journals/jar/Platzer17}, which is supported in the automated theorem prover \KeYmaeraX~\cite{DBLP:conf/cade/FultonMQVP15}. 
 For the scope of this paper, the following informal understanding, summarized in \rref{tab:dlinformalsemantics} suffices. 
 Assignment of a variable $\passign{x}{\theta}$ instantly assigns the value of $\theta$ to $x$, nondeterministic assignment $\passign{x}{\ast}$ assigns any real value to $x$ (e.g. choosing an arbitrary  acceleration value). 
 Continuous evolution $\pevolvein{\D{x} = \theta}{Q}$ evolves $x$ along the differential equation $\D{x}=\theta$ for an arbitrary duration while $Q$ is true throughout (e.g. change position of cars along kinematic equations, but prevent moving backwards). 
 The test $\ptest{P}$ checks that formula $P$ is true and continues running if it is, and if not aborts the program (possibly backtracking to other nondeterministic options).
 Sequential composition $\alpha\texttt{;} \beta$ first runs hybrid program $\alpha$ and subsequently runs hybrid program $\beta$ from the resulting states of $\alpha$ (e.g. first choose acceleration, then follow motion). 
 A nondeterministic choice of operations allows for the choice of two separate hybrid programs, $\pchoice{\alpha}{\beta}$ (e.g. running $\alpha$ if there is a safe distance between two cars and running emergency braking $\beta$ otherwise). 
 Nondeterministic repetition $\prepeat{\alpha}$ repeats hybrid program $\alpha$ some $n \geq 0$ times (e.g. repeatedly running control and motion). 
  
 \begin{table}[tb]
\caption{Syntax and informal semantics of hybrid programs}
\label{tab:dlinformalsemantics}
\begin{tabularx}{\linewidth}{lX}
 \toprule
 Program & Meaning \\
 \midrule
 $\passign{x}{\theta}$ & Assigns the value of $\theta$ to $x$\\
 $\passign{x}{*}$ & Assigns any real value to $x$ \\
 $\pevolvein{\D{x} = \theta}{Q}$ & Evolves $x$ along $\theta$ for any duration $t \geq 0$ while $Q$ is true throughout\\
 $\ptest{P}$ &  Tests formula $P$ and continues if true, aborts if false\\
 $\alpha;\beta$ & Follows program $\alpha$ and subsequently follows program $\beta$ \\
 $\pchoice{\alpha}{\beta}$ & Follows either program $\alpha$ or program $\beta$, nondeterministically  \\
 $\prepeat{\alpha}$ & Repeats program $\alpha$ a nondeterministic $n \geq 0$ number of times \\
 \bottomrule
\end{tabularx}
\end{table}

 The set of \dL formulas is generated by the following grammar, where $\sim \,\in \{<,\leq,=,\neq,\geq,>\}$ and terms $f,g$ in $+,-,\cdot,/$ are over the reals, contain number literals, variables and $n$-ary rigid function symbols, plus support interpreted minimum, maximum, absolute value and trigonometric functions \cite{10.1007/978-3-031-10769-6_42}:
 \[ P,Q \coloneqq f \sim g \mid P \land Q \mid P \lor Q \mid P \limply Q \mid P \lbisubjunct Q \mid \forall x\,P \mid \exists x\,P \mid \dbox{\alpha}{P} \mid \didia{\alpha}{P} \enspace .\]

We use $f \sim g \sim h$ as shorthand for $f \sim g \land g \sim h$ where $f, g$.
Arithmetic operations, comparisons, Boolean connectives, and quantifiers are as usual.
Formulas of the form $\dbox{\alpha}{P}$ and $\didia{\alpha}{P}$ express properties about programs: $\dbox{\alpha}{P}$ is true when in all states reachable through transitions of program $\alpha$ the property $P$ is true (used for safety); $\didia{\alpha}{P}$ is true when there exists a state reachable by the program $\alpha$ in which $P$ is true (used for liveness, optimality, and runtime monitoring).

The semantics of \dL~\cite{DBLP:journals/jar/Platzer17} is a Kripke semantics in which the states of the Kripke model are the states of the hybrid system.
A state is a map $\omega: \vars \rightarrow \reals$, assigning a real value $\omega(x)$ to each variable $x \in \vars$ in the set of variables $\vars$.
We write $\truthset{P}$ to denote the set of states in which formula $P$ is true, $\omega \in \truthset{P}$ or equivalently $\omega \models P$ if formula $P$ is true at state $\omega$, $\models P$ for $P$ being valid (true in all states), and $\valuein{\omega}{e}$ to denote the real value of term \(e\) in state \(\omega\).
The semantics of hybrid programs is expressed as a transition relation $\reachset{\alpha} \subseteq \mathcal{S} \times \mathcal{S}$ of initial and final states in $\mathcal{S}$ \cite{DBLP:journals/jar/Platzer17} and we write $\irelmodels{\omega}{\nu}{\alpha}$ to express that program $\alpha$ reaches state $\nu$ when starting in state $\omega$.
\dL is decidable relative to either an oracle for the continuous first-order logic of differential equations or an oracle for the discrete fragment of \dL \cite[Theorem 11]{DBLP:conf/lics/Platzer12b}.

 \subsection{Differential Refinement Logic}

 Differential refinement logic \dRL \cite{loos_platzer_2016} defines a relationship between two hybrid programs $\alpha \leq \beta$ that tells us that all the states reachable from program $\alpha$ are also reachable by operations of $\beta$. 
 This allows us to transfer safety properties between programs: safety properties about a hybrid program also hold for all refinements of this program. 
 In this paper, we refine the RSS control envelopes to create deterministic controllers, which are guaranteed to inherit the verified safety properties of the original control envelopes.

 \subsection{ModelPlex} 
 A proof of $\dbox{\alpha}{P}$ gives us confidence in the safety of the model $\alpha$. 
 The remaining question is now whether $\alpha$ represents a useful real (driving) behavior, and whether an implementation is faithful to the conditions of the verified model. 
ModelPlex \cite{mitsch_modelplex_2016} combines an offline safety proof of the shape \(A \limply \dbox{\alpha}{P}\) with a runtime monitor checking whether two concrete states \(\stold,\stnew\) are connected by a program \(\alpha\), i.e., semantically whether \(\irelmodels{\stold}{\stnew}{\alpha}\).
The safety proof witnesses that, starting from states satisfying \(A\), \emph{all} states reachable by model \(\alpha\) satisfy \(P\), while satisfying the runtime monitor witnesses that the two concrete states \(\stold,\stnew\) are connected by the program \(\alpha\), and so state \(\nu\) inherits the safety proof, i.e., \(\imodels{\stnew}{P}\) if \(\imodels{\stold}{A}\).
The semantic runtime monitor is equivalently phrased in \dL as a monitor specification \(\didia{\alpha}{\bigwedge_{x\in \boundvars{\alpha}}(x=x^+)}\) with fresh variables \(x^+\) not occurring in \(\alpha\) \cite{mitsch_modelplex_2016}.
The \dL monitor specification allows ModelPlex, in contrast to online reachability analysis, to shift computation offline by using theorem proving in order to translate a hybrid systems model into a formula \(\phi \in \mathcal{M}\) where $\mathcal{M}$ is the set of quantifier-free, modality-free formulas over real arithmetic.
Note that the \dL monitor specification is not provable offline, since it introduces fresh variables \(x^+\).
Instead, the proof can be finished at runtime for two concrete states (a state \(\stold\) providing values for \(x\) and a state \(\stnew\) providing values for \(x^+\)) by plugging in concrete measurements for all variables of the runtime monitor.
When a runtime monitor \(\phi\) is satisfied over states \(\stold,\stnew\), we write \(\irelmodels{\stold}{\stnew}{\phi}\).

 \section{Formalization of RSS Safety and Optimality}
 \label{sec:rssformal}
In this section we begin our case study on RSS. We first formalize the RSS guidelines and verify the safety of the model. 
We then propose a way to verify that the RSS model drives optimally. 
After this we create deterministic controllers which we show are refinements of the general verified model and utilize these to illustrate RSS behavior in simulation.

\subsection{Formalization and Verification of RSS Safety}
Models in the subsequent sections follow the same structure. 
To capture this common structure, we introduce a safety template in \rref{model:safetyProofTemplate}. 
All subsequent safety proofs for the longitudinal motion regardless of direction (same or opposite direction of motion), follow this template, but with specific formulas and programs to fill in the placeholders for the safety distance \texttt{safeDist}, evolution domain constraint \texttt{edc}, free driving program \texttt{freeDriving}, and evasion maneuver program \texttt{properResponse}. 

\setcounter{modelline}{0}
\begin{model}[tb]
\caption{Model template for longitudinal motion collision avoidance \(\texttt{longitudinal}(\texttt{safeDist},\texttt{edc},\texttt{freeDriving}, \texttt{properResponse})\)}
\label{model:safetyProofTemplate}
\begin{align*}
\text{init}&\left|    
\begin{aligned}
    &\mline{line:first-init} \phantom{\land}~ x_1\leq x_2 \land \texttt{safeDist}(v_1,v_2)\leq x_1-x_2 \\
    &\mline{line:second-init} \land 0 < a_{minBrake} < a_{maxBrake} \land 0 < a_{maxAccel} \\
    &\mline{line:third-init}\land \rho>0 \\    
    &\mline{line:last-init}\land \texttt{edc}
\end{aligned}    
    \right.\\
    &\,\phantom{\bigl|}\mline{line:implication}\limply\\
                &\,\phantom{\bigl|}\mline{line:first-loopstart}\bigl[\bigl( \\
\text{ctrl} &\left|
\begin{aligned}
  &\mline{line:first-ctrl-1} \quad \phantom{\bigl(}\bigl( \ptest{\texttt{safeDist}(v_1,v_2)}\leq x_2-x_1; \ \passign{(a_1,a_2)}{\texttt{freeDriving}} \\
  &\mline{line:first-ctrl-2} \quad  \phantom{\bigl(\bigl(} \cup\,  \\
  & \mline{line:first-ctrl-3} \quad \phantom{\bigl(\bigl(} \ptest{\texttt{safeDist}(v_1,v_2)} \geq x_2-x_1; \  \passign{(a_1,a_2)}{\texttt{properResponse}}\bigl);\\
  & \mline{line:first-ctrl-3} \quad \phantom{\bigl(} \passign{t}{0}; 
 \end{aligned}
 \right.\\
\text{motion}   &\left|
\begin{aligned}
    &\mline{line:first-ode} \phantom{\bigl(} \quad \pevolvein{\D{x_1}=v_1 \syssep \D{x_2}=v_2\syssep\D{v_1}=a_1\syssep \D{v_2}=a_2 \syssep \D{t}=1}{\texttt{edc} \land t \leq \rho}
\end{aligned}
\right.\\
\text{safe} &\,\Bigl|\mline{line:first-safe} \phantom{\Bigl[\,}\bigr)^\ast\bigr]\left(x_1 \leq x_2 \right)
\end{align*}
\end{model}

The template defines preconditions (\rref{line:first-init}--\ref{line:last-init}) under which all runs of the hybrid program in \rref{line:first-loopstart}--\ref{line:first-ode} are expected to establish the safety condition in \rref{line:first-safe}. 
The preconditions \texttt{init} set parameter bounds for acceleration, deceleration, and reaction time $\rho$, make sure that the cars are oriented correctly (driving in the expected direction and car 1 is closer to the lane origin than car 2), and are at a safe distance initially. 
The $\texttt{ctrl}$ program implements the RSS control envelope by setting the accelerations of the cars based on whether the distance between them is safe: the nondeterministic choice in \rref{line:first-ctrl-2} expresses the two options of executing \texttt{freeDriving} when at a safe distance and evasion maneuver \texttt{properResponse} otherwise (we use notation $\passign{(a_1,a_2)}{\texttt{prg}}$ to emphasize that the program \texttt{prg} chooses acceleration values $a_1$ and $a_2$). 
At the boundary, either choice is possible. 
The $\texttt{motion}$ program uses a differential equation to model the continuous dynamics of driving in response to the chosen acceleration, modeled through kinematics. 
Subsequent sections utilize this template to create specific \dL models to prove several aspects of RSS safety. 
First we start with cars driving in the same longitudinal direction. 

\subsection{Longitudinal Safety}
In the longitudinal direction, two cars may either drive in the same direction (a follower car following at a safe distance behind a lead car), or in opposite directions (a narrow road with two-way traffic, or a lead car backing up). 
We describe both models as instances of the template \rref{model:safetyProofTemplate}, using the sign of their respective longitudinal velocities to determine driving direction. 
We analyze these driving situations in isolation, under the assumption that a gear change when fully stopped is necessary to transition between these situations, which also allows for switching between the formal models that guide runtime safety.

\paragraph{Same Longitudinal Direction.}
For two cars driving in the same longitudinal direction, RSS \cite{shalev-shwartz_formal_2018} defines a safe distance between these cars with respect to their velocities: if the cars were to drive with worst-case behavior (rear car with the maximum possible acceleration and lead car with the maximum possible deceleration) for the entire reaction time $\rho$, there would still be enough space for the rear car to come to a full stop before the stopping point of the lead car, i.e., before colliding, see \eqref{eq:rsssamedirsafe}.

\begin{equation}\label{eq:rsssamedirsafe}
\texttt{safeDist}_s \triangleq
\max(v_1\rho + \frac{1}{2}\amaxAccel + \frac{(v_1 + \rho \amaxAccel)^2}{2\aminBrake} - \frac{v_2^2}{2\amaxBrake}, 0)
\end{equation}

Whenever the cars are not satisfying the safe distance \eqref{eq:rsssamedirsafe}, the cars must execute the proper response behavior as an evasion maneuver. 
In a leader and follower setup, the follower car is responsible for not hitting the lead car. 
As a consequence, the proper response allows the lead car to continue with any acceleration or deceleration that does not exceed maximum braking, but requires the follower car to decelerate at rate at least $\amin$ \cite{shalev-shwartz_formal_2018}.
We will now fill in \rref{model:safetyProofTemplate} with formulas and programs \eqref{eq:rsssamefreedriving}--\eqref{eq:rsssameedc}, which gives \cite[Model 1]{loos_safe_2011}.
\begin{align}    
    \texttt{freeDriving}_s &\triangleq \passign{a_1}{\ast}; \ptest{{-}\amaxBrake\leq a_1 \leq \amaxAccel}; \label{eq:rsssamefreedriving} \\ &\phantom{\triangleq}~~ \passign{a_2}{\ast}; \ptest{{-}\amaxBrake\leq a_2 \leq \amaxAccel} \notag \\
    \texttt{properResponse}_s &\triangleq \bigl(\pchoice{\passign{a_1}{\ast}; \ptest{a_1\leq -\aminBrake}}{\ptest{v_1=0}; \passign{a_1}{0}}\bigl); \label{eq:rsssameproperresponse} \\
    &\phantom{\triangleq}~~ \bigl( \pchoice{\passign{a_2}{*}; \ptest{a_2\geq \amaxBrake}}{\ptest{v_2=0}; \passign{a_2}{0}}\bigl) \notag \\
    \texttt{edc}_s &\triangleq  v_1 \geq 0 \land v_2 \geq 0 \label{eq:rsssameedc}
\end{align}
The definition of $\texttt{properResponse}$ includes a choice of $\passign{a_{1, 2}}{0}$ if $v_{1, 2} = 0$ in order to allow time to pass in the ODE if one or both cars are stopped.

\begin{theorem}[Same Longitudinal Safety]\label{thm:rsssame}
\rref{model:safetyProofTemplate} with \texttt{safeDist} as per \eqref{eq:rsssamedirsafe}, \texttt{freeDriving} as per \eqref{eq:rsssamefreedriving}, \texttt{properResponse} as per \eqref{eq:rsssameproperresponse}, and \texttt{edc} as per \eqref{eq:rsssameedc} is valid.
\end{theorem}
\begin{proof}
The key insight to this proof is the loop invariant $J \equiv x_1 \leq x_2 \land x_1 + \frac{v_1^2}{2\amin} \leq x_2 + \frac{v_1^2}{2\amaxb}$ for proving inductive safety using the loop invariant (LI) rule. 
$$\infer[LI]{\Gamma \vdash [\alpha*]P, \Delta}{\Gamma \vdash  J, \Delta & J \vdash P & J \vdash [\alpha]J }$$
The loop invariant $J$ summarizes that the cars maintain sufficient distance to stop before colliding. 
The resulting real arithmetic proof obligations become tractable by relating terms of the safe distance formula to terms in the solution of the differential equations using a series of cuts. 
The details of this proof are in \rref{app:casestudy}.
\qed
\end{proof}

\paragraph{Opposite Longitudinal Direction.}
The responsibility for collision avoidance is shared among the cars when they drive towards each other (in opposite longitudinal direction on the same lane).
Similar to the same longitudinal direction, the safe distance to be maintained in opposite direction driving considers the worst-case behavior for the full reaction time $\rho$ followed by the collision avoidance proper response. 
In opposite direction, the worst-case behavior is accelerating towards each other with acceleration of $\amax$ for $a_1$ and $-\amax$ for $a_2$ in the opposing direction.
For two cars driving with opposite longitudinal velocity signs, let $v_{1, \rho} = v1 + \rho \amaxAccel$ and $v_{2, \rho} = |v_2| + \rho \amaxAccel$. 
Then, \cite{shalev-shwartz_formal_2018} defines the safe distance in longitudinal opposite direction as in \eqref{eq:rssoppsafedist}.
\begin{equation}\label{eq:rssoppsafedist}
\texttt{safeDist}_o = \frac{v_1 + v_{1, \rho}}{2}\rho + \frac{(v_{1, \rho})^2}{2\aminBrake} + \frac{v_2 + v_{2, \rho}}{2} \rho + \frac{(v_{2, \rho})^2}{2\aminBrake}
\end{equation}
In this case, we assume the cars drive towards each other (car 1 driving ``forward'', i.e., away from the lane origin, car 2 driving ``backward'', i.e., towards the lane origin), and want both of them to decelerate in the proper response. 
This differs from the same longitudinal direction case, where only the follower car reacts. 
In this case, the car driving forward reacts to imminent collision with deceleration $a_1 \leq -\amin$, and the car driving backward reacts with deceleration $a_2 \geq \amin$, formalized in \eqref{eq:rssoppfreedriving}--\eqref{eq:rssoppedc}.
 \begin{align}
     \texttt{freeDriving}_o &\triangleq \passign{a_1}{\ast}; \ptest{-\amaxBrake\leq a_1 \leq \amaxAccel}; \label{eq:rssoppfreedriving} \\ &\phantom{\triangleq}~~ \passign{a_2}{\ast}; \ptest{-\amaxAccel\leq a_2 \leq \amaxBrake}\notag \\
     \texttt{properResponse}_o &\triangleq \bigl(\pchoice{(\passign{a_1}{\ast}; \ptest{a_1\leq -\aminBrake})} {(\ptest{v_1=0}; \passign{a_1}{0})}\bigl); \label{eq:rssoppproperresponse} \\
     &\phantom{\triangleq}~~ \bigl(\pchoice{(\passign{a_2}{\ast}; \ptest{a_2\geq \aminBrake})}{(\ptest{v_2=0}; \passign{a_2}{0})}\bigl) \notag \\
\texttt{edc}_o &\triangleq  v_1 \geq 0 \land v_2 \leq 0 \label{eq:rssoppedc}
 \end{align}

\begin{theorem}[Opposite Longitudinal Safety]\label{thm:rssopp}
\rref{model:safetyProofTemplate} with $\texttt{safeDist}_o$ \eqref{eq:rssoppsafedist}, $\texttt{freeDriving}_o$ \eqref{eq:rssoppfreedriving}, $\texttt{properResponse}_o$  \eqref{eq:rssoppproperresponse}, and $\texttt{edc}_o$ \eqref{eq:rssoppedc} is valid.
\end{theorem}
\begin{proof}
The proof of this model follows a similar structure to the proof of \rref{thm:rsssame} in the same longitudinal direction. The key insight again is a loop invariant: $J \equiv x_1\leq x_2 \land x_2 - \frac{v_2^2}{2\amin} \geq x_1 + \frac{v_1^2}{2\amin}$.
Details of this proof are in \rref{app:casestudy}.
\qed
\end{proof}

\subsection{Formalization and Verification of RSS Optimality}

This section analyzes whether the safe distance is the minimal safe distance needed to avoid collision.  
To this end, we encode the rationale behind $\texttt{safeDist}_s$ and $\texttt{safeDist}_o$ as a hybrid program, and ask if we were to follow a smaller safe distance, would there exist collisions (which is a liveness property of the shape $P \limply \didia{\alpha}Q$).
Recall that in both cases, the safe distance was determined by one iteration of worst-case behavior of duration $\rho$, followed by executing the proper response until both cars are stopped.
\rref{model:minDistSame} captures this intuition behind the safe distance formulas in a hybrid program: 
first, the worst-case behavior in \rref{line:mindist-ctrl-1}--\ref{line:mindist-ode-1} is executed, followed by any number and duration of proper response in \rref{line:mindist-proper-1}--\ref{line:mindist-ode-2}.
We show that starting from an initial separation of $\texttt{safeDist}-\varepsilon$ for any arbitrarily small $\varepsilon$, leads to collisions, i.e., we can reach states in which $x_1>x_2$.

\setcounter{modelline}{0}
\begin{model}[tbhp]
\caption{Minimum Distance Same Direction}
\label{model:minDistSame}
\begin{align*}
\text{init}    &\left| \begin{aligned} &\mline{line:mindist-init} \phantom{\land} \texttt{init}~\text{per \rref{model:safetyProofTemplate}} \\
    &\mline{line:mindist-epsilon}\land \varepsilon > 0 \\
    &\mline{line:mindist-safedist}\land \texttt{safeDist}(v_1, v_2) - \varepsilon = x2-x1\limply\\
\end{aligned}\right.\\
&\,\phantom{\bigl|}\mline{line:mindist-loopstart}\bigl<\bigl(\\
\text{worst case} &\left|
\begin{aligned}
  &\mline{line:mindist-ctrl-1} \quad \phantom{\bigl(}\texttt{freeDriving};\\
  &\mline{line:mindist-ctrl-2} \quad \phantom{\bigl(}\passign{t}{0};\\
&\mline{line:mindist-ode-1} \phantom{\bigl(} \quad \pevolvein{\D{x_1}=v_1 \syssep \D{x_2}=v_2\syssep \D{v_1}=a_1\syssep \D{v_2}=a_2 \syssep \D{t}=1}{t \leq \rho \land \texttt{edc}}
 \end{aligned}
 \right.\\
\text{attempt evade}   &\left|
\begin{aligned}
    &\mline{line:mindist-proper-1}\phantom{\bigl\langle\bigl(~} \bigl(\texttt{properResponse;}\\
    &\mline{line:mindist-ode-2} \phantom{\bigl\langle\bigl(\bigl(~} \pevolvein{\D{x_1}=v_1 \syssep \D{x_2}=v_2\syssep  \D{v_1}=a_1\syssep \D{v_2}=a_2 \syssep \D{t}=1}{\texttt{edc}}\bigl)^\ast
\end{aligned}
\right.\\
\text{unsafe} &\,\Bigl|\mline{line:mindist-safe} \phantom{\Bigl[\,}\bigr)\bigr>\left(x_1 > x_2 \right)
\end{align*}
\end{model}

\begin{theorem}[Same Longitudinal Optimality]
\rref{model:minDistSame} with \texttt{safeDist} as per \eqref{eq:rsssamedirsafe}, \texttt{freeDriving} as per \eqref{eq:rsssamefreedriving}, \texttt{properResponse} as per \eqref{eq:rsssameproperresponse}, and \texttt{edc} as per \eqref{eq:rsssameedc} is valid.
\end{theorem}

\begin{proof}

The proof strategy for this case simulates worst case behavior, choosing $a_1=\amaxAccel$, $a_2=-\aminBrake$. 
Unrolling the loop twice enables both cars to come to a complete stop and stay stopped by picking $a_{1,2}=0$ once stopped.
The choice of accelerations and loop unrolling results in a collision.
Note that two iterations are necessary in the case that one car stops before the other one does, in order to allow for the stopped car to change its deceleration to $a_{1,2}=0$ so that the differential equation can be followed without violating the evolution domain constraint $v_{1,2} \geq 0$ until the second car stops.
The proof strategy again simplifies arithmetic proof obligations by matching terms in the safe distance formula with terms in the solution of the differential equation. 
Since the distance between the cars is initially equal to safe distance reduced by an arbitrarily small $\varepsilon$, using \rref{thm:rsssame} we conclude that the minimum safe distance between the two cars is characterized by \eqref{eq:rsssamedirsafe}. 
\qed
\end{proof}

\begin{theorem}[Opposite Longitudinal Optimality]
\rref{model:minDistSame} with \texttt{safeDist} as per \eqref{eq:rssoppsafedist}, \texttt{freeDriving} as per \eqref{eq:rssoppfreedriving}, \texttt{properResponse} as per \eqref{eq:rssoppproperresponse}, and \texttt{edc} as per \eqref{eq:rssoppedc} is valid.
\end{theorem}
\begin{proof}
The proof strategy follows Lemma 3: we again choose worst case behaviors of both cars, namely both cars driving toward each other at maximum acceleration. 
\qed
\end{proof}

\section{Refinement from \dL to Python}

The previous section proves correctness of control decisions in all situations using \dL under an assumed model of car dynamics. 
Increased model fidelity often comes at the price of higher verification effort.
For example, if we wanted to consider tire friction, the dynamical motion model would no longer have symbolic closed-form solutions, but would require advanced differential equation reasoning \cite{DBLP:journals/jacm/PlatzerT20}.
In order to determine whether a model is sufficiently realistic, and to test concrete implementation behavior in specific situations (e.g., how abruptly a car would brake when encountering obstacles), we now discuss refinement to deterministic implementation and simulation.
This provides a way of showing fidelity evidence for models to complement the safety evidence, and to debug unverified models.

In formal models, nondeterminism is beneficial to model a large variety of concrete implementations at an abstract level: nondeterministic choice of acceleration in the controllers, nondeterministic number of loop iterations, and nondeterministic duration of ODEs within each loop iteration capture variation in total driving time and behavior. 
For execution and testing purposes, deterministic behavior is beneficial to produce repeatable, predictable, and comprehensible results.
In order to generate control code, integrate with testing, and ensure that a simulation of a system follows its formal model, we select a deterministic subset of hybrid programs to introduce an intermediate language that can then be directly compiled to Python.
This language, \dethp, supports assignment, deterministic choice (defined from nondeterministic choice and tests), and loops (defined from nondeterministic repetition and tests). 
The deterministic subset \dethp of hybrid programs is generated by the grammar below, where $\theta$ is an arithmetic term in $+,-,\cdot,/$ and $P$ is a quantifier-free formula in real arithmetic:
\[ \alpha,\beta \coloneqq \passign{x}{\theta} \mid \alpha;\beta \mid \underbrace{\pchoice{(\ptest{P};\alpha)}{(\ptest{\lnot P};\beta)}}_{\text{if}~P~\alpha~\text{else}~\beta} \mid \underbrace{\prepeat{(\ptest{P};\alpha)};\ptest{\lnot P}}_{\text{while}~P~\alpha} \]

Starting from the formal models in \rref{sec:rssformal}, we need a way of resolving the following points of nondeterminism when examining the behavior of these models:
First, the nondeterministic choice between free driving and proper response at the boundary of the safe distance region needs to be resolved (e.g., aggressively by favoring \texttt{freeDriving}, or conservatively by favoring \texttt{properResponse}). 
Second, the nondeterministic choice of acceleration value in both \texttt{freeDriving} and \texttt{properResponse} must be resolved with a deterministic computation of acceleration within the acceleration bounds (e.g., a bang-bang controller implementation will only choose the values of $\amax, \amin, $ and $\amaxb$). 
Third, the nondeterministic duration of differential equations and hence the nondeterministic control cycle time must be resolved (e.g., conservatively by picking the maximum reaction time $\rho$, but not violating the remaining conditions of the evolution domain constraint).
Fourth, the nondeterministic number of repetitions must be resolved (e.g., by picking a test length or a finite amount of time to be followed). 
Note that the model parameter choices and initial values for variables are not sources of nondeterminism as they are symbolically defined and preset. 

We now illustrate \dethp with a deterministic controller \texttt{freeDriving-det} that chooses the same acceleration for every iteration of $\texttt{freeDriving}$ below:

    \[\texttt{freeDriving-det} \triangleq \passign{a_1}{\amax}; \passign{a_2}{-\amaxb}\]
    
In order to inherit \rref{thm:rsssame}, we use \dRL to prove that \texttt{freeDriving-det} reaches a subset of the states of $\texttt{freeDriving}_s$ \eqref{eq:rsssamefreedriving}. 
\begin{theorem}[Refinement]
    The deterministic free driving controller refines the RSS free driving control envelope, i.e., $\texttt{freeDriving-det} \leq \texttt{freeDriving}_s$ \eqref{eq:rsssamefreedriving} is valid under assumptions $\Gamma \equiv \amax>0 \land \amaxb > 0$.
\end{theorem}
\begin{proof}
\newcommand{\accbounds}[1]{-\amaxBrake\leq #1 \leq \amaxAccel}
\newcommand{\assignone}[0]{\passign{a_1}{\amax}}
\newcommand{\assigntwo}[0]{\passign{a_2}{-\amaxb}}
\newcommand{\assignthree}[0]{\passign{a_1}{\ast}}
\newcommand{\assignfour}[0]{\passign{a_2}{\ast}}
We prove by refinement in \dRL.
Let $P \equiv \accbounds{a_1}$ and $Q \equiv \accbounds{a_2}$:  
{
\footnotesize
\[
\infer[;]{ \Gamma \vdash \assignone;\assigntwo  \leq \assignthree;\ptest{P}; \assignfour;\ptest{Q}}{\textbf{B1} \qquad \infer[:=]{ \Gamma \vdash [\assignone;\assigntwo]\ptest{\ltrue} \leq P}{\infer[?]{\Gamma, a_1=\amaxAccel, a_2=-\amaxBrake \vdash \ptest{\ltrue} \leq P}{\infer[\reals]{\Gamma, a_1=\amaxAccel, a_2=-\amaxBrake \vdash \ltrue \rightarrow P}{\ast}}}}
\]
\[
\infer{\textbf{B1}}{\infer[;]{\Gamma \vdash \{\assignone;\assigntwo \} \leq \{ \assignthree;\assignfour \}}{
\infer[*]{\Gamma \vdash \assignone \leq \assignthree}{\ast} \ \ \infer[:=]{\Gamma \vdash [\assignone](\assigntwo \leq \assignfour)}{\infer[*]{\Gamma, a_1=\amaxAccel \vdash \assigntwo \leq \assignfour}{\ast}}
}}
\]
}

In this proof we utilize the test, sequential and assignment rules \cite{loos_platzer_2016} to show that each step that the deterministic controller takes is within the bounds that the nondeterministic controller sets. 
The main insight is that deterministic assignments are refinements of unguarded nondeterministic assignments.
\qed
\end{proof}

The intermediate language \texttt{det-HP} serves as a stepping stone towards compilation to Python. We introduce the compilation rules in \rref{fig:compilationrules}, where 
we use the following subset of the Python expression and statement syntax \cite{kohl_2020} (for simplicity, we will ignore the difference between $\reals$ and Float, which can be addressed with interval arithmetic~\cite{DBLP:conf/pldi/BohrerTMMP18}):
\newcommand{\pexpr}{\ensuremath\text{pexpr}}
\newcommand{\pprg}{\ensuremath\text{pprg}}
\begin{align*}
\pexpr & \coloneqq \text{true} \mid \text{false} \mid z \in \mathbb{Z} \mid x \in \text{Float} \\
& \phantom{\coloneqq} \mid -\pexpr \mid \pexpr \circ \pexpr \mid \pexpr \bullet \pexpr\\
& \phantom{\coloneqq} \mid \texttt{not}~\pexpr \mid \pexpr~\texttt{and}~\pexpr \mid \pexpr~\texttt{or}~\pexpr\\
\pprg &\coloneqq \texttt{x=}\pexpr \mid \texttt{if}~\pexpr\text{:}~\pprg~\texttt{else:}~\pprg \mid \texttt{while}~\pexpr\texttt{:}~\pprg \mid \pprg;\pprg
\end{align*}

The simplified syntax uses arithmetic operators $\circ \in \{\texttt{+},\texttt{-},\texttt{*},\texttt{/},\texttt{**}\}$ and comparison operators $\bullet \in \{\texttt{<},\texttt{<=},\texttt{>=},\texttt{>},\texttt{==},\texttt{!=}\}$.
\begin{figure}[tb]
{
\begin{align*}
\textbf{Programs}\\
\pycomp{\passign{x}{\theta}} &\vartriangleright \pycomp{x}\texttt{=}\pycomp{\theta}\\
\pycomp{\pchoice{\ptest{P};\alpha}{\ptest{\lnot P};\beta}} &\vartriangleright \texttt{if}~\pycomp{P}\texttt{:}~\pycomp{\alpha}~\texttt{else:}~\pycomp{\beta}\\
\pycomp{\prepeat{\{\ptest{P};\alpha\}};\ptest{\lnot P}} &\vartriangleright \texttt{while}~\pycomp{P}:~ \pycomp{\alpha}~\texttt{else pass}\\
\textbf{Arithmetic}\\
\pycomp{z} & \vartriangleright \texttt{z} && \text{for number literal $z$}\\
\pycomp{x} & \vartriangleright \texttt{x} && \text{for variable $x$}\\
\pycomp{\theta+\eta} &\vartriangleright \pycomp{\theta}\texttt{+}\pycomp{\eta}\\
\pycomp{\theta-\eta} &\vartriangleright \pycomp{\theta}\texttt{-}\pycomp{\eta}\\
\pycomp{\theta \cdot \eta} &\vartriangleright \pycomp{\theta}\texttt{*}\pycomp{\eta}\\
\pycomp{\frac{\theta}{\eta}} &\vartriangleright \pycomp{\theta}\texttt{/}\pycomp{\eta}\\
\pycomp{\theta^\eta} &\vartriangleright \pycomp{\theta}\texttt{**}\pycomp{\eta}\\
\textbf{Boolean}\\
\pycomp{\ltrue} &\vartriangleright \texttt{True}\\
\pycomp{\lfalse} &\vartriangleright \texttt{False}\\
\pycomp{\theta < \eta} & \vartriangleright \pycomp{\theta}\texttt{<}\pycomp{\eta} && \text{similar for}~\leq,=,\neq,\geq,>\\
\pycomp{\lnot P} &\vartriangleright \texttt{not}~\pycomp{P}\\
\pycomp{P \land Q} & \vartriangleright \pycomp{P}~\texttt{and}~\pycomp{Q}\\
\pycomp{P \lor Q} &\vartriangleright \pycomp{P}~\texttt{or}~\pycomp{Q}
\end{align*}
}%
\caption{$\texttt{det-HP}$ to Python compilation rules}
\label{fig:compilationrules}
\end{figure}
The compilation from deterministic hybrid programs to the restricted Python syntax preserves safety, see Lemmas \ref{lem:terms}--\ref{lem:formulas} and \rref{thm:compilationcorrectness}.
We abbreviate arithmetic evaluation \cite{kohl_2020} as $\pyeval{\pycomp{\theta}}{\nu} \rightarrow u$ to real value $u$ and Boolean evaluation as $\pybooleval{\nu}{\pycomp{\theta}} \rightarrow v$ to Boolean value $v$.

\begin{theoremEnd}[end,restate,category=compilation,text link=]{lemma}[Term compilation is correct]\label{lem:terms}
Assuming Float=$\reals$, terms evaluate equivalently, i.e., if $\pyeval{\pycomp{\theta}}{\nu} \rightarrow u$ then $\valuein{\nu}{\theta} = u$.
\end{theoremEnd}
\begin{proof}
    By structural induction over \dL terms, see \rref{app:compilation}.
    \qed
\end{proof}
\begin{proofEnd}
    Structural induction over \dL terms as follows:
    \begin{description}
    \item[Number literals] $\pyeval{n}{\nu} \rightarrow n$ for all states $\nu$ and $\valuein{\nu}{n} = n$
    \item[Variables] $\pyeval{v}{\nu} \rightarrow \nu(v)$ and \dL variable valuation $\valuein{\nu}{v} = \nu(v)$
    \item[Negation] $\pyeval{-\theta}{\nu} \rightarrow -c$ for $\pyeval{\theta}{\nu}=c$ and \dL negation $\valuein{\nu}{-\theta}=-\valuein{\nu}{\theta}$
    \item[Arithmetic operator] for terms $\theta,\eta$ and arithmetic operator $\sim$, in Python $\pyeval{\theta \sim \eta}{\nu} \rightarrow c \sim d$ for $\pyeval{\theta}{\nu} \rightarrow c$ and $\pyeval{\eta}{\nu} \rightarrow d$ and \dL $\valuein{\nu}{\theta \sim \eta}=\valuein{\nu}{\theta} \sim \valuein{\nu}{\eta}$ \qed
    \end{description}
\end{proofEnd}

\begin{theoremEnd}[end,restate,category=compilation,text link=]{lemma}[Formula compilation is correct]\label{lem:formulas}
Formulas evaluate equivalently, i.e., $\nu \models \phi$ iff $\pybooleval{\nu}{\pycomp{\phi}} \rightarrow \texttt{true}$.
\end{theoremEnd}
\begin{proof}
    By structural induction over \dL formulas, see \rref{app:compilation}.
    \qed
\end{proof}
\begin{proofEnd}
    Structural induction over \dL formulas as follows:
    \begin{description}
    \item[Comparisons] for terms $\theta,\eta$ and comparison operators $\sim \in \{<,\leq,=,\neq,\geq,>\}$, in Python $\pybooleval{\nu}{\pycomp{\theta < \eta}} \rightarrow c~\texttt{<}~d$ (accordingly for $\leq,=\neq,\geq,>$) for $\pyeval{\pycomp{\theta}}{\nu} \rightarrow c$ and $\pyeval{\pycomp{\eta}}{\nu} \rightarrow d$ and in \dL $\valuein{\nu}{\theta \sim \eta}=\valuein{\nu}{\theta} \sim \valuein{\nu}{\eta}$
    \item[Logical connectives] for formulas $\phi,\psi$ and operators $\lnot,\land,\lor$, in Python\linebreak $\pybooleval{\nu}{\pycomp{\phi \land \psi}} \rightarrow p~\texttt{and}~q$ for $\pybooleval{\nu}{\pycomp{\phi}} \rightarrow p$ and $\pybooleval{\nu}{\pycomp{\psi}} \rightarrow q$ and in \dL $\valuein{\nu}{\phi \land \psi}=\valuein{\nu}{\phi} \land \valuein{\nu}{\psi}$, accordingly for $\lnot,\lor$ \qed
    \end{description}
\end{proofEnd}

\begin{theoremEnd}[end,restate,category=compilation,text link=]{thm}[Compilation of \dethp to Python is correct]\label{thm:compilationcorrectness}
    All states reachable with the compiled Python program are also reachable by the source HP program, i.e. 
    $(\pycomp{\alpha}, \nu) \rightarrow (\perp, \omega)$ then $\irelmodels{\nu}{\omega}{\alpha}$.
\end{theoremEnd}
\begin{proof}
By structural induction over hybrid programs, see \rref{app:compilation}.
\qed
\end{proof}
\begin{proofEnd}
By structural induction over hybrid programs. 
Recall that states in Python \cite{kohl_2020} have the form $\nu = (S,H,M)$ (each mapping variables to values), which for simplicity we restrict to programs operating only on the stack. 
As a result, we consider states in Python $S: \vars \rightarrow \reals$ the same as states in \dL mapping variables to reals $\nu: \vars \rightarrow \reals$. 
The induction base case is $(\perp, \nu)$ so $\irelmodels{\nu}{\nu}{\ptest{\ltrue}}$ and induction hypothesis $(\pycomp{\alpha}, \nu) \rightarrow (\perp, \omega) $ then $\irelmodels{\nu}{\omega}{\alpha}$.

\begin{description}
    \item[If statement $\pycomp{\bm{\pchoice{\ptest{P};\alpha}{\ptest{\lnot P};\beta}}} \vartriangleright \texttt{if}~\pycomp{\bm{P}}\texttt{:}~\pycomp{\bm{\alpha}}~\texttt{else:}~\pycomp{\bm{\beta}}$]
    Let\linebreak $(\texttt{if}~\pycomp{P}\texttt{:} \pycomp{\alpha}~\texttt{else:}~\pycomp{\beta},\nu) \rightarrow (\perp, \omega)$.
    We have to show that $\irelmodels{\nu}{\omega}{\pchoice{\ptest{P};\alpha}{\ptest{\lnot P};\beta}}$.
    \begin{description}
    \item[Case $\bm{\nu \models P}$:]
    We have $\pybooleval{\nu}{\pycomp{P}} \xrightarrow{} \texttt{true}$ by \pyrule{bool-eval} and \pyrule{bool-true} \cite[p. 83]{kohl_2020} and thus $\irelmodels{\nu}{\nu}{\ptest{P}}$ by \rref{lem:formulas} and \dL.
    By rule \pyrule{if-true} therefore $(\texttt{if}~\pycomp{P}\texttt{:}~\pycomp{\alpha}~  \texttt{else:}~\pycomp{\beta},\nu) \rightarrow (\pycomp{\alpha},\nu)$ and in turn $(\pycomp{\alpha},\nu) \rightarrow (\perp, \omega)$. 
    Now $\irelmodels{\nu}{\omega}{\alpha}$ by I.H. and thus $\irelmodels{\nu}{\omega}{\ptest{P};\alpha}$ by \dL.
    
    \item[Case $\bm{\nu \not\models P}$:]
    We have $\pybooleval{\nu}{\pycomp{P}} \xrightarrow{} \texttt{false}$ by \pyrule{bool-eval} and \pyrule{bool-false} and thus $\irelmodels{\nu}{\nu}{\ptest{\lnot P}}$ by \rref{lem:formulas} and \dL.
    By rule $\pyrule{if-false}$ therefore $(\texttt{if}~\pycomp{P}~\pycomp{\alpha}~\texttt{else:}\pycomp{\beta},\nu) \rightarrow (\pycomp{\beta},\nu)$ and in turn $(\pycomp{\beta},\nu) \rightarrow (\perp, \omega)$. 
    Now $\irelmodels{\nu}{\omega}{\beta}$ by I.H. and thus $\irelmodels{\nu}{\omega}{\ptest{\neg P}; \beta}$ by \dL.    
    \end{description}
    Now we have that either $\irelmodels{\nu}{\omega}{\ptest{P};\alpha}$ or  $\irelmodels{\nu}{\omega}{\ptest{\lnot P};\beta}$ and thus we conclude $\irelmodels{\nu}{\omega}{\pchoice{\ptest{P};\alpha}{\ptest{\lnot P};\beta}}$ by \dL. 
    
    \item[Sequential composition $\pycomp{\bm{\alpha;\beta}}\vartriangleright\pycomp{\bm{\alpha}}\texttt{;}\pycomp{\bm{\beta}}$] 
    We have to show $\irelmodels{\nu}{\omega}{\alpha;\beta}$. 
    Let $(\pycomp{\alpha;\beta}, \nu) \rightarrow (\perp, \omega)$. 
    By \pyrule{sequence-first} and \pyrule{sequence-elim} \cite[p. 91]{kohl_2020}, there exists state $\mu$ such that $(\pycomp{\alpha}, \nu) \rightarrow (\perp, \mu)$ and $(\pycomp{\beta}, \mu) \rightarrow (\perp, \omega)$. 
    This means that $\irelmodels{\nu}{\mu}{\alpha}$ and $\irelmodels{\mu}{\omega}{\beta}$ by I.H and we conclude $\irelmodels{\nu}{\omega}{\alpha;\beta}$ by \dL.
    \item[Assignment $\pycomp{\bm{\passign{x}{\theta}}} \vartriangleright \pycomp{\bm{x}}\bm{=}\pycomp{\bm{\theta}}$] 
    We have to show that $\irelmodels{\nu}{\omega}{\passign{x}{\theta}}$. 
    From \texttt{assign-exec} we know $(\pycomp{x}=\pycomp{\theta},\nu) \rightarrow (\perp, \nu[x \mapsto u])$ for $\pyeval{\pycomp{\theta}}{\nu} \rightarrow u$. 
    As a result, there is some state $\omega$ that agrees with all values of $\nu$ except for the value of $\pycomp{x}$, so $\omega(\pycomp{x}) = \valuein{\nu}{\pycomp{\theta}}$ and we conclude $\irelmodels{\nu}{\omega}{\passign{x}{\theta}}$ by \dL from \rref{lem:terms}.
    
    \item[While loop $\pycomp{\bm{\prepeat{(\ptest{P};\alpha)};\ptest{\neg P}}} \vartriangleright \texttt{\bf while}~\pycomp{\bm{P}}:~\pycomp{\bm{\alpha}}~\texttt{\bf else pass}$] 
    We have to show that $\irelmodels{\nu}{\omega}{\prepeat{(\ptest{P};\alpha)};\ptest{\lnot P}}$. 
    Let \[(\texttt{while}~\pycomp{P}:~\pycomp{\alpha}~\texttt{else pass}, \nu) \rightarrow (\perp, \omega) \enspace .\]
    We prove by induction on the loop iterations.    
    The induction base case is $\nu \models \lnot P$ and therefore $(\texttt{while}~\pycomp{P}:~\pycomp{\alpha}~\texttt{else pass}, \nu) \rightarrow (\perp, \nu)$ by \pyrule{while-eval} \cite[p. 94]{kohl_2020}, \pyrule{while-false} \cite[p. 94]{kohl_2020}, and \pyrule{pass} \cite[p. 92]{kohl_2020} and in turn $\irelmodels{\nu}{\nu}{\ptest{\ltrue}}$.
    Assume $(\texttt{while}~\pycomp{P}\texttt{:}~\pycomp{\alpha}~\texttt{else pass},\nu) \rightarrow (\perp,\mu)$ then $\irelmodels{\nu}{\mu}{\prepeat{(\ptest{P};\alpha)};\ptest{\lnot P}}$.
    \begin{description}
    \item[Case $\bm{\mu \models P}$] 
    We have $\pybooleval{\mu}{\texttt{true}}$ by \pyrule{while-eval}. 
    Then, 
    \begin{multline*}
    (\texttt{while}~\pycomp{P}\texttt{:}~\pycomp{\alpha}~\texttt{else pass},\mu) \rightarrow\\ 
    (\pycomp{\alpha};\texttt{while}~\pycomp{P}\texttt{:}~\pycomp{\alpha}~\texttt{else pass},\mu)
    \end{multline*} 
    by \pyrule{while-true}. 
    By sequential composition there exists state $\tilde{\mu}$ such that $(\pycomp{\alpha},\mu) \rightarrow (\perp,\tilde{\mu})$ and $(\texttt{while}~\pycomp{P}\texttt{:}~\pycomp{\alpha}~\texttt{else pass},\tilde{\mu}) \rightarrow (\perp,\omega)$.
    Now $\irelmodels{\mu}{\tilde{\mu}}{\alpha}$ by outer I.H. and $\irelmodels{\tilde{\mu}}{\omega}{\prepeat{(\ptest{P};\alpha)};\ptest{\lnot P}}$ by inner I.H. and we conclude $\irelmodels{\nu}{\omega}{\prepeat{(\ptest{P};\alpha)};\ptest{\lnot P}}$ by \dL.

    \item[Case $\bm{\mu \not\models P}$]
    Then $(\texttt{while}~\pycomp{P}:~\pycomp{\alpha}~\texttt{else pass},\mu) \rightarrow (\texttt{pass},\mu)$ by \pyrule{while-eval} and \pyrule{while-false} and in turn $(\texttt{pass},\mu) \rightarrow (\perp,\omega)$ by \pyrule{pass}. 
    Now $\irelmodels{\mu}{\omega}{\prepeat{(\ptest{P};\alpha)};\ptest{\lnot P}}$ by outer I.H. and we conclude $\irelmodels{\nu}{\omega}{\prepeat{(\ptest{P};\alpha)};\ptest{\lnot P}}$ by inner I.H. \qed
    \end{description}
\end{description}
\end{proofEnd}

\rref{thm:compilationcorrectness} ensures that any state reached in a compiled Python program will also be reachable in the source $\texttt{det-HP}$ program. 
Together with the refinement step from our original model, we therefore know that any counterexample found in the compiled simulation will also be a counterexample of the original formal \dL model.

\section{Testing and Monitoring by Example}
\label{sec:testing}

The previous sections introduced safety verification and refinement proofs that eliminate the need for correctness testing. 
These two steps, if completed, assure that the system is safe under all implementations. 
In this section, we consider testing and simulation as a method to provide insight for correcting the formal model and/or an implementation when either of these steps fails, and as a way to assess the fidelity of the formal model and experiment with model parameter selection once the proofs are completed.

\subsection{Monitor Structure}
In order to assess whether the simulation and the formal model agree and flag differences, we utilize ModelPlex~\cite{mitsch_modelplex_2016}. 
ModelPlex monitors are formally verified and derived automatically from a hybrid systems model; they flag whenever states encountered at runtime (e.g., in simulation) disagree with the expectations of the formal model.
For this, the monitors collect the relationship between starting and final states along all possible execution paths of the formal model.
For example, a monitor derived from the discrete statements of \rref{model:safetyProofTemplate} as instantiated in \rref{thm:rsssame} is illustrated below (variables of the initial state are $x,y,z$, variables of the final state are $x^+,y^+,z^+$).
\begin{align*}
     & (\texttt{safeDist}_s(v_1,v_2) \leq x_2-x_1 \land -\amaxBrake \leq a_1^+ \leq \amaxAccel \\
     & \qquad \land -\amaxBrake \leq a_2^+ \leq \amaxAccel \land t^+=0 \land \texttt{edc}_s)\\
\lor & (\texttt{safeDist}_s(v_1,v_2) \geq x_2-x_1 \land (a_1^+ \leq -\aminBrake \lor (v_1=0 \land a_1^+=0))\\ 
& \qquad \land (a_2^+ \geq \amaxBrake \lor (v_2=0 \land a_2^+=0)) \land t^+=0 \land \texttt{edc}_s)
\end{align*}

\subsection{Examples of Monitor Usage}

\paragraph{Monitoring for faulty implementation.} 
Our first example illustrates the use of testing and monitoring when the refinement proof cannot be completed due to an incorrect implementation. 
With user-defined test scenarios, ModelPlex monitors in simulation highlight the vulnerabilities of the controller implementation as an intermediate step towards finding implementation fixes and completing the refinement proof. 
In this case, counterexamples found in simulation represent control decisions of the unverified controller implementation that violate the expectations of the formal model. 
\begin{figure}[t!]
    \centering
    \includegraphics[width=.45\textwidth]{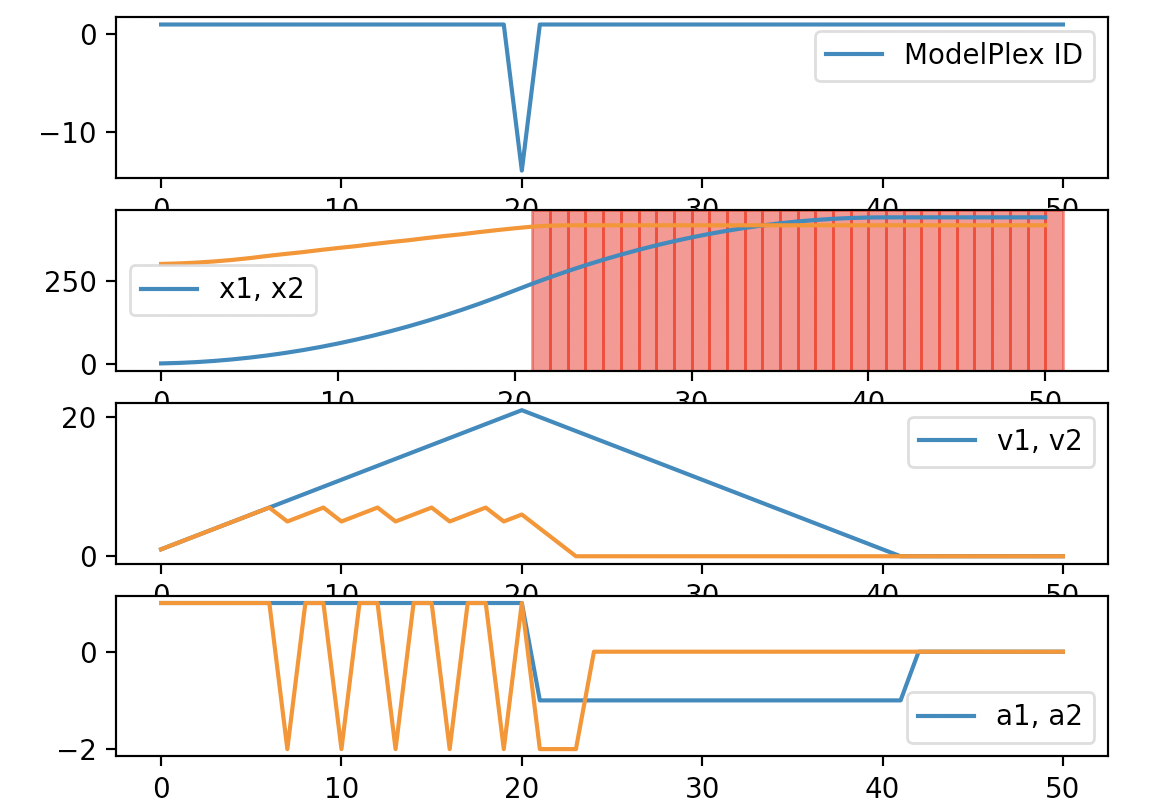}
    \caption{Leader-follower same direction test scenario, with rear car control implementation violating the conditions of RSS. Plots from top to bottom are: ModelPlex monitor verdict, car positions $x_1,x_2$, car speeds $v_1,v_2$, and accelerations $a_1,a_2$.}
    \label{fig:samedirectionplot3}
\end{figure}
In \rref{fig:samedirectionplot3}, a front and a rear car are driving in the same direction, while the rear car follows a faulty controller. 
The user-defined test scenario is an example of a boundary test by picking worst case behavior where the rear car is accelerating at maximum acceleration whenever allowed by its control guards. 
The ModelPlex monitor derived from \rref{model:safetyProofTemplate} with the formulas and programs of \rref{thm:rsssame} fails at the point of violating RSS: executing the proper response instead of the faulty controller at time time of the monitor alarm would have prevented collision, but following the faulty controller makes the rear car exceed the front car's position (i.e., cause a collision). 
The monitor verdict (``ModelPlex ID'' in \rref{fig:samedirectionplot3}) flags the violated condition of the monitor, which directly relates to a path through the formal model: the faulty controller allowed to continue accelerating forward for one more time step when it should have been following the proper response already.
In this case, a boundary test reveals the controller implementation flaw, but other implementation bugs are more subtle and require extensive search.
For future work, we envision integration of monitors with falsification techniques.

\paragraph{Monitoring for modeling flaws.}
With a setup similar to monitoring for faulty implementations, we address modeling flaws in unverified models. 
User-defined test scenarios allow us to examine the behavior of the formal model through monitors: the formal model is faulty when finding scenarios that end in collision without a prior ModelPlex monitor alarm.
In this case, the insight for fixing the bug is not merely the (missing) monitor alarm itself, but the sequence of control actions that lead up to the collision.

\paragraph{Parameter selection.}
The last example we include illustrates how to use tests for parameter selection once the formal model and the refinement to a concrete controller implementation is proved.
The concrete choices for the symbolic parameters, such as $\amaxBrake$, $\aminBrake$, $\amaxAccel$, and $\rho$, heavily influence the specific behavior of the controller implementation.
Simulation with specific parameter choices is a useful technique to inspect the expected real-world behavior, and completes the circle of interaction between formal methods and testing. 
Formal methods guarantee correctness in a model of physics, refinement proofs and correct compilation guarantee correctness of implementations, and simulation and testing allow for inspecting the fidelity of the formal model and the concrete behavior of the implementation under specific parameter choices.
\begin{figure}[b!]
    \centering
    \subfigure{\includegraphics[width=.49\textwidth]{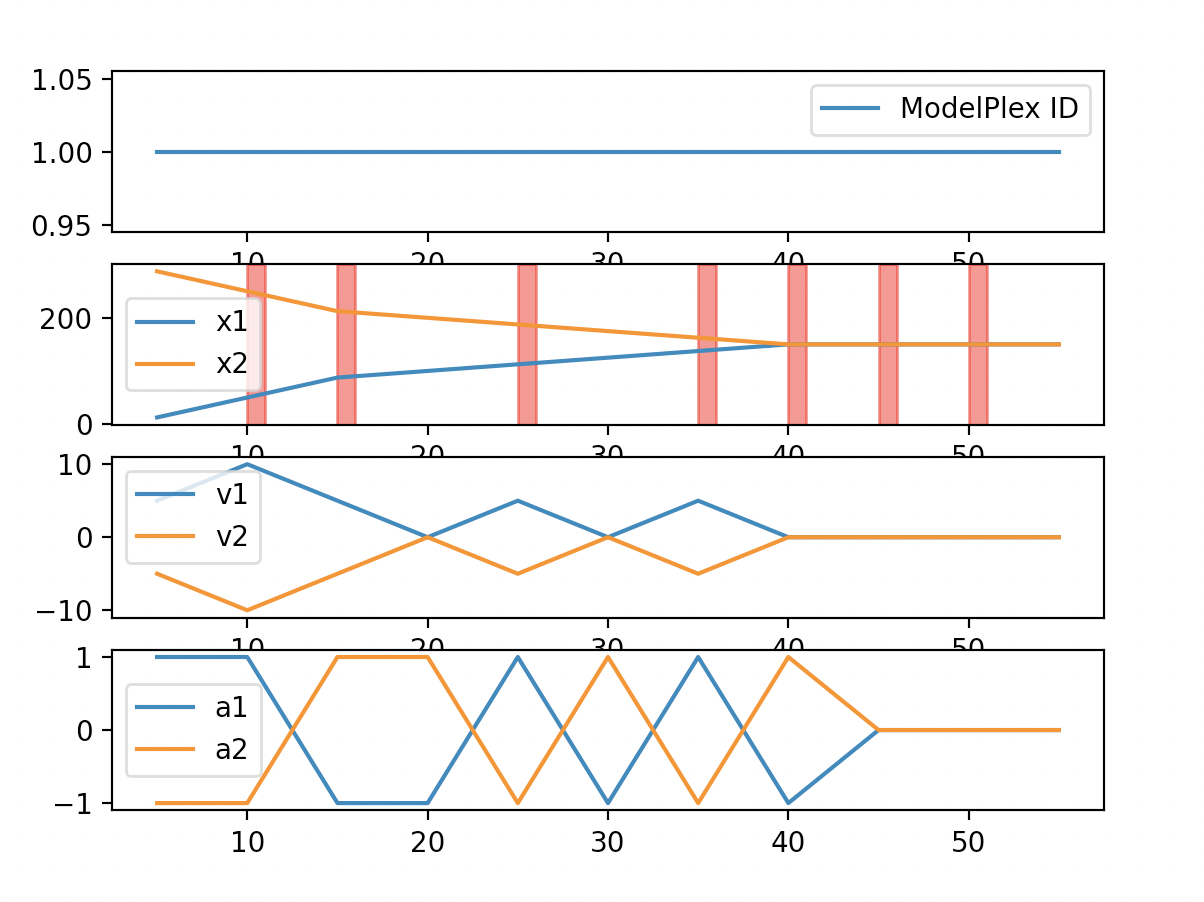}}
    \subfigure{\includegraphics[width=.49\textwidth]{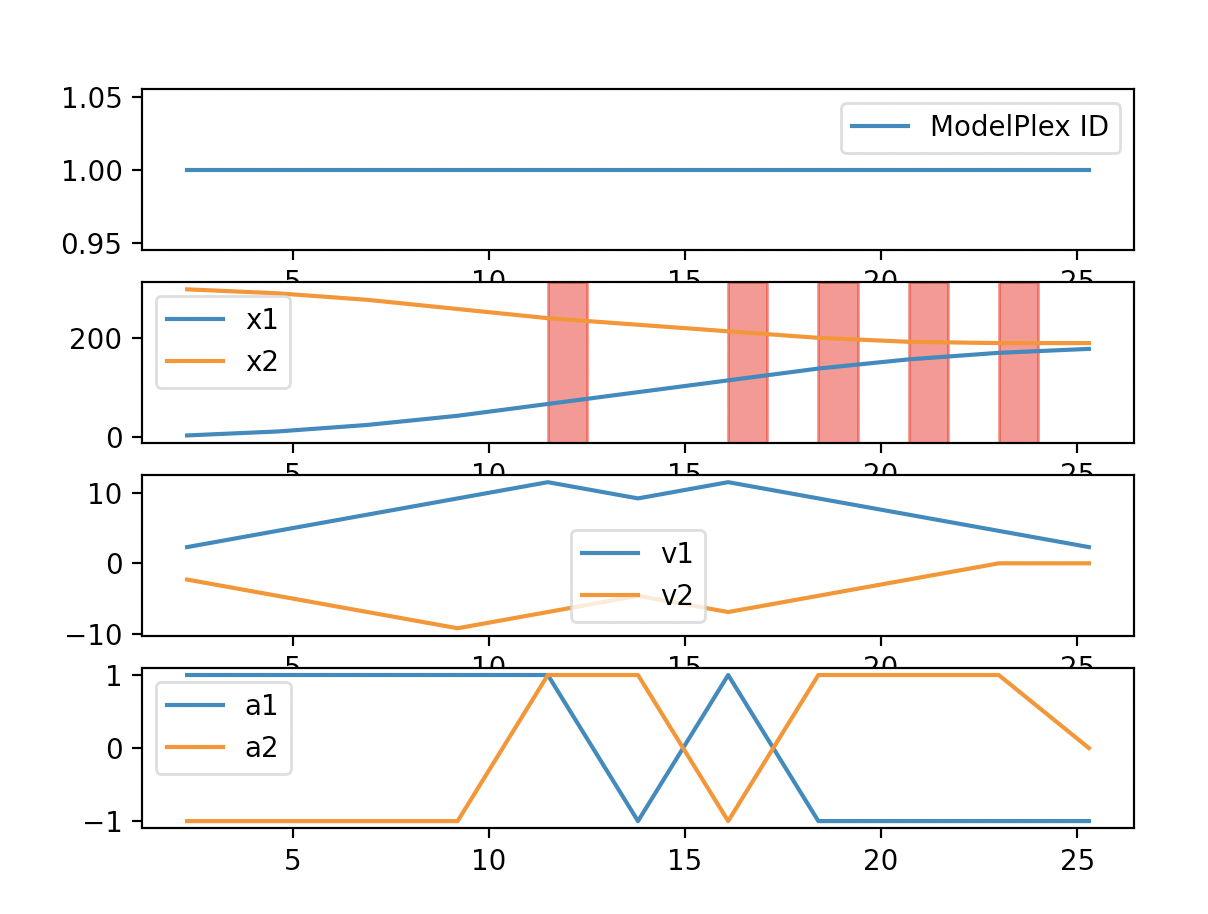}}
    \caption{Car behavior in opposite direction model at $\rho=1$ (left) and $\rho = 6$ (right).}
    \label{fig:samedirectionplot}
\end{figure}
In \rref{fig:samedirectionplot} we consider parameter selection and its effect on the behavior of cars in the opposite direction scenario. 
The controllers of both cars in free driving accelerate towards each other at maximum acceleration and in proper response they both brake at minimum rate. 
With a ``short'' response horizon of $\rho=1$, the final distance between the cars once fully stopped is much lower than when a response $\rho=6$. 

\section{Related Work}
RSS introduced in \cite{shalev-shwartz_formal_2018} has been deployed by the Mobileye research group as the basis for their self-driving car and analyzed extensively from different aspects for its safety and reliablity: \cite{naumann_responsibility_2021} considers parameter analysis and concrete behavior of RSS on German highways; \cite{koopman} extends RSS fidelity by applying Newtonian mechanics to analyze car behavior on slopes and in varying road conditions.

Related work within formal verification of driving models is extensive: \cite{noauthor_towards_nodate} utilize a formula equivalent to the same direction safe distance formula introduced by RSS for an adaptive cruise control system, \cite{DBLP:journals/ijrr/MitschGVP17} analyze collision avoidance while following Dubin's paths, \cite{DBLP:journals/ral/BohrerTMSP19} and \cite{DBLP:conf/fm/KopylovMNW21} formally verify waypoint navigation, \cite{DBLP:conf/amcc/AbhishekSJ20} verify swerving maneuvers to avoid collision, \cite{DBLP:conf/atva/RizaldiISA18} develops a formally verified motion planner, \cite{7995802,DBLP:conf/itsc/LiuWKA22} use reachability analysis in highway settings and introduce a toolbox to create simulated driving environments to determine reachable sets, \cite{DBLP:journals/tiv/KlischatA23} is close to our work in that it uses falsification to find mistakes in motion plans of autonomous vehicles.
We complement these approaches by providing proofs of optimality of the RSS safety conditions and proofs of opposite driving direction, combined with refinement proofs and correct compilation to Python for testing.

\section{Conclusion and Future Work}
In this paper we formalized and proved safety and optimality of the Responsibility-Sensitive Safety control envelopes in both longitudinal directions. 
The safety proofs discover loop invariants, which concisely document design properties of the system. 
These properties of the longitudinal cases do not transfer to the lateral case in a straightforward way, because velocity direction changes do not require an intermediate stop such as a gear change; formalizing and verifying the lateral case is an area of future research. 

To complement the safety guarantees of formal verification with practical tools to test and debug unverified models, as well as to assess the fidelity of verified models in specific driving situations and under specific parameter choices, we introduce refinement proofs and correct compilation from nondeterministic models to deterministic implementations in Python.
Formal methods guarantee correctness in a model of physics, refinement proofs and correct compilation guarantee correctness of implementations, and simulation and testing allow for inspecting the fidelity of the formal model and the concrete behavior of the implementation under specific parameter choices.
Future research includes combination with falsification methods to provide a systematic automated approach to testing of unverified models in differential dynamic logic.

\subsubsection*{Acknowledgements}
This material is based upon work supported by the National Science Foundation under Grant No. CCF2220311.

\newpage
\bibliographystyle{splncs04}
\bibliography{main}

\begin{thebibliography}{10}
\providecommand{\url}[1]{\texttt{#1}}
\providecommand{\urlprefix}{URL }
\providecommand{\doi}[1]{https://doi.org/#1}

\bibitem{DBLP:conf/amcc/AbhishekSJ20}
Abhishek, A., Sood, H., Jeannin, J.: Formal verification of swerving maneuvers
  for car collision avoidance. In: 2020 American Control Conference, {ACC}
  2020, Denver, CO, USA, July 1-3, 2020. pp. 4729--4736 (2020).
  \doi{10.23919/ACC45564.2020.9147679}

\bibitem{7995802}
Althoff, M., Koschi, M., Manzinger, S.: Commonroad: Composable benchmarks for
  motion planning on roads. In: 2017 IEEE Intelligent Vehicles Symposium (IV).
  pp. 719--726 (2017). \doi{10.1109/IVS.2017.7995802}

\bibitem{DBLP:conf/pldi/BohrerTMMP18}
Bohrer, R., Tan, Y.K., Mitsch, S., Myreen, M.O., Platzer, A.: Veriphy: verified
  controller executables from verified cyber-physical system models. In: PLDI.
  pp. 617--630 (2018). \doi{10.1145/3192366.3192406}

\bibitem{DBLP:journals/ral/BohrerTMSP19}
Bohrer, R., Tan, Y.K., Mitsch, S., Sogokon, A., Platzer, A.: A formal safety
  net for waypoint-following in ground robots. {IEEE} Robotics Autom. Lett.
  \textbf{4}(3),  2910--2917 (2019). \doi{10.1109/LRA.2019.2923099}

\bibitem{davenport_real_1988}
Davenport, J.H., Heintz, J.: Real quantifier elimination is doubly exponential
  \textbf{5}(1),  29--35. \doi{10.1016/S0747-7171(88)80004-X}

\bibitem{DBLP:conf/cade/FultonMQVP15}
Fulton, N., Mitsch, S., Quesel, J., V{\"{o}}lp, M., Platzer, A.: Keymaera {X:}
  an axiomatic tactical theorem prover for hybrid systems. In: Automated
  Deduction - {CADE-25} - 25th International Conference on Automated Deduction,
  Berlin, Germany, August 1-7, 2015, Proceedings. pp. 527--538 (2015).
  \doi{10.1007/978-3-319-21401-6\_36}

\bibitem{10.1007/978-3-031-10769-6_42}
Gallicchio, J., Tan, Y.K., Mitsch, S., Platzer, A.: Implicit definitions with
  differential equations for keymaera x. In: Blanchette, J., Kov{\'a}cs, L.,
  Pattinson, D. (eds.) Automated Reasoning. pp. 723--733. Springer
  International Publishing, Cham (2022)

\bibitem{DBLP:journals/tiv/KlischatA23}
Klischat, M., Althoff, M.: Falsifying motion plans of autonomous vehicles with
  abstractly specified traffic scenarios. {IEEE} Trans. Intell. Veh.
  \textbf{8}(2),  1717--1730 (2023). \doi{10.1109/TIV.2022.3191179}

\bibitem{kohl_2020}
Kohl, M.A.: An Executable Structural Operational Formal Semantics for Python.
  Ph.D. thesis (2020)

\bibitem{koopman}
Koopman, P., Osyk, B., Weast, J.: Autonomous vehicles meet the physical world:
  Rss, variability, uncertainty, and proving safety. In: Romanovsky, A.,
  Troubitsyna, E., Bitsch, F. (eds.) Computer Safety, Reliability, and
  Security. pp. 245--253. Springer International Publishing, Cham (2019)

\bibitem{DBLP:conf/fm/KopylovMNW21}
Kopylov, A., Mitsch, S., Nogin, A., Warren, M.: Formally verified safety net
  for waypoint navigation neural network controllers. In: Formal Methods - 24th
  International Symposium, {FM} 2021, Virtual Event, November 20-26, 2021,
  Proceedings. pp. 122--141 (2021). \doi{10.1007/978-3-030-90870-6\_7}

\bibitem{DBLP:conf/itsc/LiuWKA22}
Liu, E.I., W{\"{u}}rsching, G., Klischat, M., Althoff, M.: Commonroad-reach:
  {A} toolbox for reachability analysis of automated vehicles. In: 25th {IEEE}
  International Conference on Intelligent Transportation Systems, {ITSC} 2022,
  Macau, China, October 8-12, 2022. pp. 2313--2320 (2022).
  \doi{10.1109/ITSC55140.2022.9922232}

\bibitem{loos_safe_2011}
Loos, S.M., Platzer, A.: Safe intersections: At the crossing of hybrid systems
  and verification. In: 2011 14th International {IEEE} Conference on
  Intelligent Transportation Systems ({ITSC}). pp. 1181--1186. {IEEE}.
  \doi{10.1109/ITSC.2011.6083138}

\bibitem{loos_platzer_2016}
Loos, S.M., Platzer, A.: Differential refinement logic. In: Proceedings of the
  31st Annual {ACM/IEEE} Symposium on Logic in Computer Science, {LICS} '16,
  New York, NY, USA, July 5-8, 2016. pp. 505--514 (2016).
  \doi{10.1145/2933575.2934555}

\bibitem{DBLP:journals/ijrr/MitschGVP17}
Mitsch, S., Ghorbal, K., Vogelbacher, D., Platzer, A.: Formal verification of
  obstacle avoidance and navigation of ground robots. Int. J. Robotics Res.
  \textbf{36}(12),  1312--1340 (2017). \doi{10.1177/0278364917733549}

\bibitem{noauthor_towards_nodate}
Mitsch, S., Loos, S.M., Platzer, A.: Towards formal verification of freeway
  traffic control. In: 2012 {IEEE/ACM} Third International Conference on
  Cyber-Physical Systems, {ICCPS} 2012, Beijing, China, April 17-19, 2012. pp.
  171--180 (2012). \doi{10.1109/ICCPS.2012.25}

\bibitem{mitsch_modelplex_2016}
Mitsch, S., Platzer, A.: {ModelPlex}: verified runtime validation of verified
  cyber-physical system models  \textbf{49}(1),  33--74.
  \doi{10.1007/s10703-016-0241-z}

\bibitem{naumann_responsibility_2021}
Naumann, M., Wirth, F., Oboril, F., Scholl, K., Elli, M.S., Alvarez, I., Weast,
  J., Stiller, C.: On responsibility sensitive safety in car-following
  situations - a parameter analysis on german highways. In: 2021 {IEEE}
  Intelligent Vehicles Symposium ({IV}). pp. 83--90.
  \doi{10.1109/IV48863.2021.9575420}

\bibitem{platzer_differential_2010}
Platzer, A.: Differential dynamic logics: Automated theorem proving for hybrid
  systems  \textbf{24}(1),  75--77. \doi{10.1007/s13218-010-0014-6}

\bibitem{DBLP:conf/lics/Platzer12b}
Platzer, A.: The complete proof theory of hybrid systems. In: LICS. pp.
  541--550. IEEE (2012). \doi{10.1109/LICS.2012.64}

\bibitem{DBLP:journals/jar/Platzer17}
Platzer, A.: A complete uniform substitution calculus for differential dynamic
  logic. J. Autom. Reason.  \textbf{59}(2),  219--265 (2017).
  \doi{10.1007/s10817-016-9385-1}

\bibitem{DBLP:journals/jacm/PlatzerT20}
Platzer, A., Tan, Y.K.: Differential equation invariance axiomatization. J.
  {ACM}  \textbf{67}(1),  6:1--6:66 (2020). \doi{10.1145/3380825}

\bibitem{DBLP:conf/atva/RizaldiISA18}
Rizaldi, A., Immler, F., Sch{\"{u}}rmann, B., Althoff, M.: A formally verified
  motion planner for autonomous vehicles. In: Automated Technology for
  Verification and Analysis - 16th International Symposium, {ATVA} 2018, Los
  Angeles, CA, USA, October 7-10, 2018, Proceedings. pp. 75--90 (2018).
  \doi{10.1007/978-3-030-01090-4\_5}

\bibitem{shalev-shwartz_formal_2018}
Shalev-Shwartz, S., Shammah, S., Shashua, A.: On a formal model of safe and
  scalable self-driving cars

\bibitem{tarski_decision_1998}
Tarski, A.: A decision method for elementary algebra and geometry. In:
  Caviness, B.F., Johnson, J.R. (eds.) Quantifier Elimination and Cylindrical
  Algebraic Decomposition. pp. 24--84. Texts and Monographs in Symbolic
  Computation, Springer. \doi{10.1007/978-3-7091-9459-1\_3}

\end{thebibliography}

\newpage
\appendix

\section{Correctness of $\texttt{det-HP}$ to Python Compilation}
\label{app:compilation}

For the proof of $\texttt{det-HP}$ compilation we utilize the \dL semantics described in \cite{platzer_differential_2010}. We also use the subset in \rref{fig:pythonsemantics} of the semantic rules found in \cite{kohl_2020}.
\begin{figure}[htb]
\[\infer[\texttt{if-true}]{(\texttt{if True then a else b}) \rightarrow a}{}\]
\[\infer[\texttt{if-false}]{(\texttt{if False then a else b}) \rightarrow b}{}\]
\[\infer[\texttt{while-false}]{[\texttt{False}] c \texttt{while P: a else b} \rightarrow b}{}\]
\[\infer[\texttt{while-true}]{[\texttt{True}] c \texttt{while P: a else b} \rightarrow [\alpha]\ a\ \ \texttt{while P: a else b}}{}\]
\[\infer[\texttt{sequence-first}]{(\texttt{s ; s'}) \rightarrow (\texttt{t ; s'})}{\texttt{s ; t}}\]
\[\infer[\texttt{sequence-elim}]{(\perp \texttt{; s}) \rightarrow \texttt{s}}{}\]
\[\infer[\texttt{assign-exec}]{(\texttt{x = u}, \nu) \rightarrow (\perp, \nu[x \mapsto u])}{}\]
\caption{Semantic rules of Python \cite{kohl_2020}}
\label{fig:pythonsemantics}
\end{figure}

\printProofs[compilation]

\section{Case Study Details}
\label{app:casestudy}

\subsection{Same Longitudinal Direction Safety}
In this section we describe in detail the strategy behind proving the safety of the same longitudinal direction RSS model. 
Since the basis of this model is a loop, we utilize a loop invariant to aid in proving this. 
The invariant $J$ expresses that at each iteration of the loop in \rref{model:safetyProofTemplate}, we know two facts: (1) that car 1 stays behind car 2 and (2) that the latest stopping position of car 1 when applying $-a_{minBrake}$ deceleration is behind the earliest stopping position of car 2 when applying $-a_{maxBrake}$ deceleration. 
Formally, our loop invariant is 
$$J \equiv x_1 \leq x_2 \land x_1 + \frac{v_1^2}{2\amin} \leq x_2 + \frac{v_1^2}{2\amaxb}$$
The induction step of the proof requires us to show $J \limply \dbox{\text{\ref{model:safetyProofTemplate}}}{J}$.
Our proof proceeds by cases according to \rref{model:safetyProofTemplate}: $\texttt{safeDist}(v_1, v_2) \leq x_2 - x_1$ and $\texttt{safeDist}(v_1, v_2) \geq x_2 - x_1$. 
\\
\\
\begin{description}
\item[Case $\texttt{safeDist}(v_1, v_2) \leq x_2 - x_1$]
We need to show the free driving case
\begin{align*}
J \limply [&\ptest{\texttt{safeDist}(v_1,v_2) \leq x_2-x_1;\\
&\passign{(a_1,a_2)}{\texttt{freeDriving};\\
&\passign{t}{0};\\
&\text{motion}}}]J\enspace .
\end{align*}
In this case, we know from \texttt{freeDriving} that the chosen acceleration for both cars is within the overall bounds, which include positive and negative values between the maximum braking parameter and maximum acceleration parameter. 
We recall the fact that the differential equation can run for any amount of time such that the evolution domain constraints are satisfied, i.e., before exceeding time $\rho$. 
We solve the differential equation to obtain the usual kinematic constraints:
\begin{align*}
\tilde{x}_1&=x_1 + v_1t + \frac{1}{2}a_1t^2\\
\tilde{x}_2&=x_2 + v_2t + \frac{1}{2}a_2t^2\\
\tilde{v}_1&=v_1 + a_1t \land \tilde{v}_1 \geq 0\\
\tilde{v}_2&=v_2 + a_2t \land \tilde{v}_2 \geq 0\\
&0 \leq t \leq \rho
\end{align*}

and thus we need to show
\begin{align*}
& \phantom{\land~} \tilde{J}=x_1 + v_1t + \frac{1}{2}a_1t^2 \leq x_2 + v_2t + \frac{1}{2}a_2t^2\\
& \land
x_1 + v_1t + \frac{1}{2}a_1t^2 + \frac{(v_1 + a_1t)^2}{2\amin} \leq x_2 + v_2t + \frac{1}{2}a_2t^2 + \frac{(v_2 + a_2t)^2}{2\amaxb}
\end{align*}
In principle, since quantifier elimination for real arithmetic is decidable \cite{tarski_decision_1998,davenport_real_1988}, the proof could be finished automatically; in practice, however, the proof obligations above are too complicated to terminate in reasonable time. 
In order to split the proof into multiple more tractable subproblems, we relate terms in the loop invariant and the \texttt{safeDist} condition. 
The loop invariant $J$ follows the shape of \texttt{safeDist} closely, with $\rho$ replaced by $t$, $\amax$ replaced by $a_1$ and $\amaxb$ replaced by $a_2$. 
We consider that we substitute the most extreme cases for choice of acceleration in the safe distance formula, which is the case that the two cars are driving towards each other in the most extreme case possible, meaning that these are easily relatable. 
For the time parameter, the position functions of the two cars are monotonic, meaning that we prove a lemma that is equivalent to the safe distance formula with $\rho$ replaced with $t$. 
Specifically, we use a cut to provide the following lemma \eqref{eq:samedirlemma}: 
\begin{equation}\label{eq:samedirlemma}
\max(v_1t + \frac{1}{2}a_{maxAccel} + \frac{(v_1 + t a_{maxAccel})^2}{2a_{minBrake}} - \frac{v_2^2}{2a_{maxBrake}}, 0)\leq x2-x1
\end{equation}

Finally, with the following additional lemmas, which all follow from the fact that $0 \leq t \leq \rho$, $v1 \geq 0$, $v2\geq 0$, $-\amaxb \leq a_1 \leq \amax$, and $-\amaxb \leq a_2 \leq \amax$, we align subterms of \texttt{safeDist} and lemma \eqref{eq:samedirlemma}:

\begin{itemize}
    \item $x_1 + v_1t + \frac{1}{2}a_1t^2 \leq v_1t + \frac{1}{2}a_{maxAccel} $
    \item $ \frac{(v_1 + a_1t)^2}{2\amin} \leq \frac{(v_1 + t a_{maxAccel})^2}{2a_{minBrake}}$
    \item $\frac{(v_2 + a_2t)^2}{2\amaxb} + v_2t + \frac{1}{2}a_2t^2 \geq \frac{v_2^2}{2a_{maxBrake}} $
\end{itemize}

\item[Case $\texttt{safeDist}(v_1, v_2) \geq x_2 - x_1$]

For the proper response branch we utilize the loop invariant stopping distance constraints that we know are true at the start of the iteration. 
The stopping distance portion of the loop invariant solves for the position that each car will be at when it reaches a full stop if the $-\amin$ is applied in car 1 and car 2 stays within its maximum braking (the assumed worst case). This is exactly what the proper response does, as it assigns an acceleration to each car, where car 1 has to brake at least as hard as $-\amin$ and car 2 has to apply some acceleration greater than $-\amaxb$. 
This case closes by solution of the differential equation and quantifier elimination. \qed
\end{description}

\subsection{Opposite Longitudinal Direction Proof Details}
For the model of cars driving in opposite directions, we use a similar proof strategy to the cars driving in the same direction. We define a loop invariant to encode the information that at each iteration the cars will not pass each other, and that if both cars apply minimum braking, they will not collide. This is to say that if after an iteration of the loop, the cars would follow the proper response, they would not collide. Formally the invariant states that $$ J \equiv x_1\leq x_2 \land x_2 - \frac{v_2^2}{2\amin} \geq x_1 + \frac{v_1^2}{2\amin} $$
Consider at the end of an iteration of the loop, we get the proof obligation: 
\begin{align*}
& \phantom{\land~} x_1 + v_1t + \frac{1}{2}a_1t^2 \leq x_2 + v_2t + \frac{1}{2}a_2t^2 \\
& \land x_1 + v_1t + \frac{1}{2}a_1t^2 + \frac{(v_1 + a_1t)^2}{2\amin} \leq x_2 + v_2t - \frac{1}{2}a_2t^2 + \frac{(v_2 + a_2t)^2}{2\amin}
\end{align*}
for $0 \leq t \leq \rho$.
As in the same direction case, there is a free driving and proper response program to the model. 
In the free driving case, we relate terms in the loop invariant to terms in \texttt{safeDist} as follows: 
\begin{itemize}
    \item $v_1t + \frac{1}{2}a_1t^2 \leq v_1\rho + \frac{1}{2}a_{maxAccel} $
    \item $ v_2t + \frac{1}{2}a_1t^2 \geq v_2\rho - \frac{1}{2}a_{maxAccel} $
    \item $ \frac{(v_1 + a_1t)^2}{2\amin} \leq \frac{(v_1 + t a_{maxAccel})^2}{2a_{minBrake}}$
    \item $\frac{(v_2 + a_2t)^2}{2\amin} + v_2t + \frac{1}{2}a_2t^2 \geq \frac{v_2^2}{2a_{minBrake}} $
\end{itemize}

Using these cuts, we are able to show that the distance the cars will travel if needing to come to a complete stop will not cause a collision. Formally, this is the second term in the loop invariant. For the first term, using the fact that the distance that the front car travels is always less or equal to the distance that would be travelled at full acceleration for $\rho$ time and the distance that car 2 travels will always be greater or equal to it travelling at full negative acceleration for $\rho$ time, we are able to show that the new positions of the cars do not collide, i.e. $x_1\leq x_2$. 
The proper response case again follows by quantifier elimination automatically. \qed

\end{document}